\documentclass[
	fontsize=12pt,
	oneside,
	paper=a4,
	pagesize=auto,
	titlepage=yes,
	bibliography		= totocnumbered,
	listof					= totocnumbered,
	parskip					= half,
	abstract				= yes
	]{scrartcl}

\usepackage[utf8]{inputenc}
\usepackage[T1]{fontenc}
\usepackage[english]{babel}

\addtokomafont{descriptionlabel}{\rmfamily}
\addtokomafont{pagehead}{\scshape}

\usepackage{enumerate}
            
\usepackage{csquotes}
\usepackage[backend=biber,style=numeric,autolang=other]{biblatex}
\usepackage[headsepline]{scrlayer-scrpage}
\ohead{Fries, C.}
\ihead{Replication-Consistent Liquidity Forecasting}
\usepackage[intlimits]{amsmath}
\usepackage{amssymb,amsthm,aliascnt}	

\usepackage{graphicx,xcolor}

\definecolor{DarkBlue}{rgb}{0.0, 0.0, 0.5}
\definecolor{DarkRed}{rgb}{0.5, 0.0, 0.0}
\definecolor{DarkGreen}{rgb}{0.0, 0.5, 0.0}
\definecolor{DarkYellow}{rgb}{0.5, 0.5, 0.0}
\definecolor{Brown}{cmyk}{0.00,0.80,1.00,0.60}

\usepackage{ifpdf}	
\ifpdf
    \usepackage{hyperref}
    \hypersetup{
        pdfpagelayout=TwoPageRight,
        pdfborderstyle={/S/U/W 1}, 
        colorlinks = true,
        linktocpage = true,
        linkcolor=blue,
        citecolor=blue,
        urlcolor=blue,
        pdfauthor={Christian P. Fries},
        pdfsubject={Financial Derivatives Valuation and Risk Management},
        pdftitle={Replication-Consistent Liquidity Forecasting for Derivatives},
        pdfstartview=FitH,
        pdfkeywords={Derivative Valuation; Funding Curve; Discounting Sensitivities; Expected Cash-Flows; Liquidity Management}
    }
\else
    \usepackage{url}
\fi	
                
\usepackage{orcidlink}

\usepackage{titling}
\thanksmarkseries{arabic}

\usepackage{ifthen}
\newif\ifappendix
\appendixfalse

\newcounter{cpfNumberOfFigures} 
\newcounter{cpfNumberOfTables} 

\usepackage{listings}
\usepackage{comment}

\newcommand{\indicatorfcn}{\mathrm{\mathbf{1}}}

\newtheorem{proposition}{Proposition}

\newcommand{\articletitle}{Replication-Consistent Liquidity Forecasting for Derivatives}
\newcommand{\articlesubtitle}{Forward Funding Sensitivities and a Liquidity Valuation Adjustment for Settlement Lags}

\begin{document}

\subject{}
\author{Christian P.\ Fries\textsuperscript{\orcidlink{0000-0003-4767-2034}}
\thanks{Department of Mathematics, University of Munich, Munich, Germany}
\textsuperscript{,}%
\thanks{DZ Bank AG Deutsche Zentral-Genossenschaftsbank, Frankfurt am Main, Germany}
\textsuperscript{,}%
\thanks{Correspondence: \href{mailto:email@christian-fries.de}{email@christian-fries.de}}
}

\title{
    \articletitle \\[2ex] \large \articlesubtitle \\[2ex]
    \centerline{\small Version 0.9.7}
}
\date{September 15, 2025
    \\ {\large (revised: April 9, 2026)}
}

\maketitle

\begin{abstract}
We study cash-flow forecasting for derivatives used in liquidity management and clarify its relation to risk-neutral valuation and replication.

While it is well known that expectations under different measures (e.g., $\mathbb{P}$ vs.\ $\mathbb{Q}$) can yield different undiscounted cash-flows, further inconsistencies arise when payment times are stochastic. We show that using discounting sensitivities (funding-curve hedge ratios) instead of “expected cash-flows” aligns forecasting with the self-financing replication strategy and avoids measure-mixing/aggregation issues.

We then illustrate how a standard valuation model delivers pathwise funding requirements and propose a simple liquidity valuation adjustment to capture settlement lags and related timing frictions.

The note provides implementation hints (American Monte Carlo with adjoint differentiation) and clarifies when “expected cash-flows” are informative and when sensitivities should be used instead.
\end{abstract}

\tableofcontents

\clearpage
\section{Introduction}
\label{sec:liqui:introduction}

	Cash-flow forecasting is a central element of financial risk management in insurance and banking. For example, in the quantification of the one-year market-consistent embedded value (MCEV) distribution under Solvency~II for life insurance companies, one possible approach is the construction of static replicating portfolios that reproduce cash-flows or discounted terminal values of liabilities \cite{NatolskiWerner2014ReplicatingPortfolios}. More generally, cash-flow forecasting for non-deterministic liabilities or derivative positions may require both a model of the market factors driving contractual payments and a model of the trading activity that determines the portfolio and its funding needs over time.

	For stochastic cash-flows, expectations are a natural descriptive summary. Although expectations under the objective measure $\mathbb{P}$ may appear to be the most direct notion, they are usually difficult to estimate and are not the quantities delivered by standard pricing models. A common practical alternative is therefore to use risk-neutral valuation models to project \emph{undiscounted} expected cash-flows (or related quantities such as quantiles or exposures). For liquidity forecasting, however, this can be misleading. Undiscounted expected cash-flows depend on the chosen numeraire $N$ of the equivalent martingale measure $\mathbb{Q}^{N}$, and aggregating such expectations across products or legs valued under different numeraires may generate artificial residual terms. This measure dependence is a direct consequence of the change-of-numeraire framework of \cite{GemanElKarouiRochet1995}. Even when discounted present values are numeraire invariant, undiscounted timing profiles need not satisfy intuitive additivity properties. For products with stochastic payment times, such as Bermudan-style contracts, the sum of projected notional payments across time buckets need not agree with the total notional.

	These questions arise against the broader literature on collateral, funding, and replication in derivative valuation. The effects of collateralization and marking-to-market on intermediate cash-flows and discounting were discussed already by \cite{JohannesSundaresan2007}. Multiple-curve and collateral-aware valuation frameworks were developed, among others, in \cite{Piterbarg2010,FujiiShimadaTakahashi2009,FujiiTakahashi2013}. Replication-based treatments of funding and counterparty adjustments include \cite{BurgardKjaer2011,PallaviciniPeriniBrigo2012,Crepey2015,BieleckiRutkowski2015}; see also the discussion of fair value versus dealer-specific funding charges in \cite{HullWhite2014,AndersenDuffieSong2019}. Relatedly, replicating-portfolio methods have been used to approximate liability values and cash-flows, especially in insurance applications \cite{NatolskiWerner2014ReplicatingPortfolios,NatolskiWerner2017ReplicationByCashFlowMatching,NatolskiWerner2018Foundation}. 
	Our focus is different: we consider a derivative position that is already valued in a replication-based model and ask which model outputs are informative for liquidity management.
	In addition, delivery-versus-payment mitigates principal risk but does not eliminate liquidity risk, and recent supervisory and policy work emphasizes timing mismatches, collateral flows, and margin-related cash movements as important sources of intraday liquidity risk \cite{BIS1992DvP,ECB2024Intraday,JukonisLetiziaRousova2022}. This motivates the operational perspective adopted here: liquidity forecasting concerns the timely performance of payment obligations, while standard funding- and collateral-aware valuation frameworks assess the cost implications of different liquidity sources.

    In the nonlinear XVA/TVA literature, the term \emph{liquidity valuation adjustment} (LVA) is used in a broader pricing sense than in the present note. In~\cite{BrigoLiuPallaviciniSloth2016}, the deal value is decomposed into a clean price plus CVA, DVA, LVA and FVA, where LVA captures the costs/benefits of collateral margining and FVA the funding costs/benefits of the hedging strategy. In the TVA framework~\cite{Crepey2015}, LVA is the liquidity/funding component of the total valuation adjustment alongside CVA, DVA and replacement cost; the terminology is not fully uniform, and in the book treatment it is essentially the funding valuation adjustment net of the credit spread. By contrast, the present note uses LVA in a narrower operational sense: it measures the valuation effect of settlement lags or other timing frictions applied to replication-implied net funding cash increments, rather than a full bilateral pricing decomposition with counterparty risk, collateral and funding constraints.

	The guiding principle of the present note is therefore simple: if a product is valued by a self-financing replication strategy, liquidity-relevant cash-flow forecasts should be derived from the funding instruments' hedge ratios and their rebalancing, rather than from static undiscounted expectations of contractual legs.

	\medskip

	Against this background, the contribution of the present note is twofold. First, we show that for derivative products with stochastic payment amounts and/or times the intuitive notion of expected cash-flows can be inconsistent with replication-based liquidity management. Second, we propose a replication-consistent alternative based on funding sensitivities and their pathwise rebalancing, together with a simple settlement-lag liquidity valuation adjustment for timing frictions.

    \paragraph{Scope of the note.}
    The aim of the paper is conceptual clarification and methodological guidance. We focus on closed-form and low-dimensional examples, together with implementation guidance. A comprehensive numerical analysis is left for future work.

    \paragraph{Outline.}
    The remainder of the introduction fixes notation and presents two motivating examples: inconsistencies in projected cash-flows arising from mixing different numeraires, and the effect of ignoring the trading activity implicit in risk-neutral replication. Section~\ref{sec:liqui:expected-cashflows} then compares different notions of expected cash-flows and shows why stochastic payment times can make them misleading for liquidity forecasting. Section~\ref{sec:liqui:funding-sensitivities} investigates funding sensitivities as a replication-consistent representation of future liquidity requirements. Section~\ref{sec:liqui:replication-consistent-future-liquidity-requirement} extracts scenario-wise rebalancing increments of funding sensitivities to obtain replication-consistent liquidity requirements and introduces a settlement-lag \emph{liquidity valuation adjustment} measuring the cost of timing gaps.


\subsection{Setup and Notation}
\label{sec:liqui:intro:setup}

\paragraph{Probability space and information.}
We work on a filtered probability space
$(\Omega,\mathcal{F},(\mathcal{F}_t)_{t\ge 0},\mathbb{P})$.
All cash-flows and asset prices are adapted to $(\mathcal{F}_t)$.
We consider a finite time grid
\[
0=t_0 < t_1 < \dots < t_n,
\]
which represents the (model) re-hedging times. Contractual payment times are assumed to lie on this grid. We use $t$ for the running time and $T$ for a generic maturity, both assumed to lie on the time-grid.

\paragraph{Cash-flow streams.}
A (domestic-currency) cash-flow paid at time $t_{i}$ is represented by an $\mathcal{F}_{t_{i}}$-measurable random variable $X_{i}$.
A contract with multiple payments is represented by the stream $(X_{i})_{i=1,\ldots,n}$.
If payments occur at a stopping time $\tau\in\{t_1,\dots,t_n\}$, we will use the decomposition
\[
X_{i} \indicatorfcn_{\{\tau=t_{i}\}}, \qquad i=1,\ldots,n,
\]
where $\indicatorfcn$ denotes the indicator function.

\paragraph{Tradable numeraires and martingale measures.}
Let $N=(N(t))_{t\in[0,t_n]}$ be a strictly positive numeraire process (a tradable asset).
We assume absence of arbitrage and that prices are modeled such that there exists an equivalent martingale measure
$\mathbb{Q}^{N}$ under which all traded asset prices, discounted by $N$, are martingales.
For two numeraires $N_1,N_2$ with associated measures $\mathbb{Q}^{N_1},\mathbb{Q}^{N_2}$ we use the standard
change-of-numeraire identity
\begin{equation}
\label{eq:setup:numerairechange}
\frac{d\mathbb{Q}^{N_2}}{d\mathbb{Q}^{N_1}}\Big|_{\mathcal{F}_t}
=
\frac{N_2(t)}{N_1(t)}\,\frac{N_1(0)}{N_2(0)} \text{.}
\end{equation}

\paragraph{Risk-neutral valuation of cash-flow streams.}
For an $\mathcal{F}_{t_{i}}$-measurable payoff $X_{i}$ paid at $t_{i}$, its time-$t$ value ($t\le t_{i}$) is
\[
V_{i}(t)
=
N(t)\,\mathrm{E}^{\mathbb{Q}^{N}}\!\left(\frac{X_{i}}{N(t_{i})} \,\middle|\, \mathcal{F}_t\right).
\]
For a stream $(X_{i})_{i=1,\ldots,n}$ we write, on the grid,
\begin{equation}
\label{eq:setup:streampricing}
V(t_j)
=
N(t_j)\,\mathrm{E}^{\mathbb{Q}^{N}}\!\left(\sum_{i=j+1}^{n}\frac{X_{i}}{N(t_{i})} \,\middle|\, \mathcal{F}_{t_j}\right).
\end{equation}
By construction, the present value $V(0)$ is independent of the particular choice of numeraire $N$ (as long as the model is internally consistent and arbitrage-free), even though undiscounted expectations such as $\mathrm{E}^{\mathbb{Q}^{N}}(X_{i})$ may depend on $N$.

\paragraph{Discount factors and funding instruments.}

We denote by $P(T;t)$ the time-$t$ value of a zero-coupon bond maturing at $T$, and by $P(T)$ the corresponding value process. When distinguishing OIS discounting and a funding curve, we write
$P^{\circ}(T;t)$ for an OIS (or collateral) discount factor and $P^{f}(T;t)$ for a desk-specific funding discount factor (or funding bond). The precise curve construction is not needed for the conceptual arguments; it suffices that $P^{f}(T)$ is treated as a tradable (or model-implied) funding instrument within the valuation framework.

\paragraph{Spread/intensity notation and interpretation.}

Throughout this note we use the symbol $\lambda$ for a generic \emph{exponential spread} (attenuation rate) that enters valuations through factors of the form $\exp(-\int \lambda)$.
Depending on the application, $\lambda$ admits two interpretations:
\begin{enumerate}[(i)]
    \item as a \emph{survival/default intensity} used to model conditional payments via survival factors
    $\exp\!\left(-\int_0^{t_{i}}\lambda(u)\,du\right)$, and
    \item as a \emph{funding spread curve} defining a funding discount factor relative to an OIS (or collateral) curve, e.g.
    \[
        P^{f}(T;t)=P^{\circ}(T;t)\exp\!\left(-\int_{t}^{T}\lambda(u;t)\,du\right).
    \]
\end{enumerate}
In reduced-form credit modeling these interpretations coincide under additional assumptions
(e.g., particular recovery conventions and absence of liquidity premia), in which case $\lambda$ may be viewed as a \emph{market-implied intensity}. Otherwise, $\lambda$ should be read as a market-implied spread parameter that affects discounting in the same exponential way.

\paragraph{Cumulative spread / intensity.}
We parameterize exponential spread effects via a cumulative spread function
\begin{equation}
\label{eq:setup:cumulativespread}
\Lambda(t,T) \ := \ \int_{t}^{T} \lambda(u;t)\,du,
\end{equation}
where $\lambda(\cdot;t)$ denotes a (possibly time-$t$ dependent) instantaneous spread/intensity term structure. For grid times $t_{i}$ we write the shorthand
\[
\Lambda_{i} \ := \ \Lambda(0,t_{i}).
\]
If $\lambda$ is taken piecewise constant on the grid (as used in some discrete arguments below), then $\Lambda_{i}$ can be represented as a sum of interval averages, e.g.~$\Lambda_{i}=\sum\limits_{j=0}^{i-1}\lambda_j\,\Delta t_j$ with $\Delta t_j=t_{j+1}-t_j$.

\paragraph{Replication viewpoint and hedge ratios (sensitivities).}
We interpret the valuation process $V(t)$ as the value of a self-financing replication strategy in a chosen basis of traded instruments.
Whenever $V(t)$ is differentiable with respect to the time-$t$ price of a traded instrument $A(t)$, we denote the corresponding hedge ratio by
\begin{equation}
\label{eq:setup:hedgeratio}
\phi_{A}[V](t) \ :=\ \frac{\partial V(t)}{\partial A(t)}.
\end{equation}
In particular, for funding instruments $P^{f}(T;t)$ and (if applicable) collateral accounts $N^{C_k}(t)$ we use
\[
\phi_{P^{f}(T)}[V](t) := \frac{\partial V(t)}{\partial P^{f}(T;t)},
\qquad
\phi_{C_k}[V](t) := \frac{\partial V(t)}{\partial N^{C_k}(t)}.
\]
These hedge ratios encode the \emph{replication-consistent} decomposition of the position into funding and collateral sources.

\paragraph{Contractual cash-flows vs.\ replication-generated funding cash increments.}
A central distinction in this note is between:
\begin{itemize}
    \item \emph{contractual cash-flows} $(X_{i})$ specified by the product (paid at the contractual times), and
    \item \emph{replication-generated funding cash increments} $(\delta_{i})$ arising from rebalancing the hedge over the grid.
\end{itemize}
We represent the (scenario-wise) funding-rebalancing increment in a source $A$ over $(t_{i-1},t_{i}]$ by
\begin{equation}
\label{eq:setup:delta}
\delta_{A}[V](t_{i})\ :=\ \phi_{A}[V](t_{i})-\phi_{A}[V](t_{i-1}).
\end{equation}
In the reduced single-source setup, the increment in the shortest-maturity funding instrument is identified with the model-implied net funding cash increment over the time-step:
\begin{equation*}
    \delta_{i} \ :=\ \delta_{P^{f}(t_{i})}[V](t_{i}) \text{.}
\end{equation*}
We use the sign convention that $\delta_{i}>0$ denotes a net cash inflow available at $t_{i}$ in the frictionless baseline,
and $\delta_{i}<0$ denotes a net cash outflow to be funded at $t_{i}$.
For convenience we write $\delta_{i}^{+} := \max(\delta_{i},0)$ and $\delta_{i}^{-}:=\max(-\delta_{i},0)$.

\medskip

With this notation in place, the subsequent sections clarify why static expectations of contractual cash-flows
(e.g.\ $\mathrm{E}^{\mathbb{Q}^{N}}(X_{i})$) may be inconsistent objects for liquidity forecasting,
whereas the replication-consistent quantities $(\delta_{i})$ derived from hedge ratios provide a coherent basis for
scenario-wise funding requirements.

    \subsection{Inconsistent Expected Cash-Flows due to Measure Change}

	A common approach for determining the expectation of a future cash-flow is to use an equivalent martingale measure $\mathbb{Q}$ as an approximation for the expectation under $\mathbb{P}$. Such an approach may be justified for the projection on short time-horizons, as the change of measure is given by a drift.

	For derivative valuation it is common to use different equivalent martingale measure $\mathbb{Q}^{N}$ with different numeraires $N  =  N_{1}$, $N = N_{2}$. The motivation behind this is sometimes rooted in differences in numerical accuracy or the availability of analytic formulas.

	Using different equivalent martingale measures $\mathbb{Q}^{N_{1}}$, $\mathbb{Q}^{N_{2}}$ may give inconsistent expected cash-flows that cannot be consistently aggregated.

	\subsubsection{Example}

	Let $P(T)$ denote the zero coupon bond  maturing in time $T$.\footnote{We follow the notation of \cite{FriesLectureNotes2007}.}
	Let $L(T_{1},T_{2}) = \frac{P(T_{1})-P(T_{2})}{(T_{2}-T_{1}) P(T_{2})}$ denote the corresponding simple forward rate. Consider the expected cash-flow of a forward rate $L(T_{1},T_{2};t)$ fixed in $t = T_{1}$, paid in $T_{1}$, under the two measures $\mathbb{Q}^{P(T_{1})}$ and $\mathbb{Q}^{P(T_{2})}$. Note that the forward rate is paid at the beginning of the period.

    Using \eqref{eq:setup:numerairechange} with $N_{1}=P(T_{2})$, $N_{2}=P(T_{1})$ at time $t=T_{1}$, we obtain the corresponding expectations as
	\begin{align*}
	 	\mathrm{E}^{\mathbb{Q}^{P(T_{2})}}\left(  L(T_{1},T_{2};T_{1}) \right)
		& \ = \
		\mathrm{E}^{\mathbb{Q}^{P(T_{2})}}\left(  \frac{L(T_{1},T_{2};T_{1}) P(T_{2};T_{1})}{P(T_{2};T_{1})} \right)
		\ = \ L(T_{1},T_{2};t_{0}) 
		\intertext{and}
		\mathrm{E}^{\mathbb{Q}^{P(T_{1})}}\left(  L(T_{1},T_{2};T_{1}) \right)
		& \ = \
		\mathrm{E}^{\mathbb{Q}^{P(T_{1})}}\left(  \frac{L(T_{1},T_{2};T_{1})}{P(T_{2};T_{1})} (1 + L(T_{1},T_{2};T_{1}) (T_{2}-T_{1})) \right) \\
		& \ = \
		\mathrm{E}^{\mathbb{Q}^{P(T_{2})}}\left(  L(T_{1},T_{2};T_{1}) (1 + L(T_{1},T_{2};T_{1}) (T_{2}-T_{1})) \right) \\
		& \ = \
		\mathrm{E}^{\mathbb{Q}^{P(T_{2})}}\left(  L(T_{1},T_{2};T_{1}) \right) +  (T_{2}-T_{1}) \mathrm{E}^{\mathbb{Q}^{P(T_{2})}}\left(  ( L(T_{1},T_{2};T_{1}) )^{2} \right) \\
		& \ = \
		L(T_{1},T_{2};t_{0}) +  (T_{2}-T_{1}) \mathrm{E}^{\mathbb{Q}^{P(T_{2})}}\left(  ( L(T_{1},T_{2};T_{1}) )^{2} \right)
	\end{align*}
	If these two expectations belong to a long and a short position, a netting would leave us with the residual
	\begin{equation}
        \label{eq:liqui:measure-mix-residual}
		(T_{2}-T_{1}) \mathrm{E}^{\mathbb{Q}^{P(T_{2})}}\left(  ( L(T_{1},T_{2};T_{1}) )^{2} \right) \text{,}
	\end{equation}
	which arises solely due to the change of measure. Assuming that under the $T_{2}$-forward measure the forward rate satisfies a Black-type lognormal model
    \begin{equation*}
        L(T_{1},T_{2};T_{1}) \ = \ L_{0}\exp\!\left(-\tfrac12 \sigma_{\text{L}}^{2}T_{1}+\sigma_{\text{L}} W(T_{1}) \right)
    \end{equation*}
    the residual term \eqref{eq:liqui:measure-mix-residual} becomes $(T_{2}-T_{1}) \ L_{0}^{2} \ e^{\sigma_{\text{L}}^{2}T_{1}}$, which is exponentially increasing in variance and time horizon $T_{1}$.\footnote{For related numerical analysis on the impact of the measure choice see~\cite{ufimtsev2025liquidity}}.

    \bigskip

    \subsection{Liquidity forecasting under replication}
    
    Another inconsistency in liquidity forecasting may arise from ignoring the trading activity related to a risk-neutral replication activity. Assume a trader is managing a delta-neutral portfolio of an option $C$ and a stock $S$ being $C(t_0) - \phi(t_{0}) S(t_{0})$, where $\phi(t_{0})$ is the hedge ratio (the delta). Assessing this position at maturity $T$ in a future scenario $\omega$ may show a net cash-flow $C(T,\omega) - \phi(t_{0}) S(T,\omega)$, ignoring that the trader has made self-financing dynamic adjustments to the hedge ratio $\phi(t_{0}) \rightarrow \phi(t,\omega)$.

    \medskip

    \paragraph{Key principle (liquidity forecasting under replication):}
    If a product is valued by a self-financing replication strategy, liquidity-relevant “cash-flow forecasts” should be derived from the funding instruments’ hedge ratios and their rebalancing, not from static undiscounted expectations of contractual legs—especially when payment times are stochastic.

	\clearpage
	\section{Expected Cash-Flows}
    \label{sec:liqui:expected-cashflows}

	Let $X_{i}$ denote an $\mathcal{F}_{t_{i}}$-measurable random variable representing a cash-flow in time $t_{i}$. The natural definition of an expected cash flow is to consider the $\mathbb{P}$- or $\mathbb{Q}$-expectation:
	\begin{equation*}
		EC^{\mathbb{P}}\left( X_{i} \right) \ := \ \mathrm{E}^{\mathbb{P}}\left( X_{i} \right) \text{,} \qquad
		EC^{\mathbb{Q}}\left( X_{i} \right) \ := \ \mathrm{E}^{\mathbb{Q}}\left( X_{i} \right) \text{.}
	\end{equation*}
	For general stochastic cash-flows, these definitions depend on the probability measure.
    
    A somewhat measure-independent alternative is given by
	\begin{equation}
        \label{eq:liqui:ecadvalueperbond}
        EC^{1/P}\left( X_{i} \right) \ := \ N(t_{0}) \mathrm{E}^{\mathbb{Q}^{N}}\left( \frac{X_{i}}{N(t_{i})} \right) \cdot \frac{1}{P(t_{i};t_{0})} \text{,}
	\end{equation}
    where $N$ denotes a chosen numeraire, $\mathbb{Q}^{N}$ the corresponding equivalent martingale measure, and $P(t_{i})$ the zero coupon bond maturing in $t_{i}$. The definition \eqref{eq:liqui:ecadvalueperbond} does not depend on the numeraire, because $N(t_{0}) \mathrm{E}^{\mathbb{Q}^{N}}\left( \frac{X_{i}}{N(t_{i})} \right)$ constitute the risk-neutral valuation of the cash-flow. In other words $EC^{1/P}\left( X_{i} \right)$ is just the risk-neutral valuation of the cash-flow $X_{i}$ expressed in terms of units of the zero coupon bond  $P(t_{i};t_{0})$. With $N := P(t_{i})$ we find from $N(t_{i}) = 1$ and $N(t_{0}) =  P(t_{i};t_{0})$ that $EC^{1/P}\left( X_{i} \right) \ = \ \mathrm{E}^{\mathbb{Q}^{P(t_{i})}}\left( X_{i} \right)$.

    \medskip
    
    We may recover the definition \eqref{eq:liqui:ecadvalueperbond} as being a discounting sensitivity, i.e., a hedge ratio. This view will be important for a more general view on cash-flow forecasting.


	\subsection{Discounting Sensitivities}

	Discounting sensitivities can be re-interpreted as expected cash flows. As discounting sensitivities relate to the risk-neutral replication of the cash-flow, this relation has a fundamental relevance for liquidity management.

	We take a time-discrete view and consider a time discretization $\{ t_{i}  \ \vert \ i = 0,\ldots,n \}$.

	Consider stochastic cash-flows $X_{i}$ that are paid in $t_{i}$ and introduce an additional survival probability that models the payment of $X_{i}$ conditional to survival, where the survival probability is given through the intensity $t \mapsto \lambda(t)$.

    Let $\Delta t_{i}:=t_{i+1}-t_{i}$ and define the cumulative spread to $t_{i}$ by
    \[
    \Lambda_{i} := \Lambda(0,t_{i})=\int_{0}^{t_{i}}\lambda(u)\,du.
    \]
    
    We then have

    \begin{proposition}[Sensitivity as discounted expected Cash-Flow]
        Assume that the cash-flow does not depend on $\lambda$, i.e., $\frac{\partial X_{i}}{\partial \Lambda_{k}} = 0$. We treat the discrete curve nodes $(\Lambda_{i})_{i=1,\ldots,n}$ as independent ``bucket'' parameters (analogous to bumping discount factors by maturity), and we evaluate
        sensitivities at the baseline $\Lambda_{i}=0$.
        Then, the discounting sensitivity agrees with the (analytically) discounted expected cash-flow under terminal measure, i.e., 
        \begin{equation*}
            - \frac{\partial}{\partial \Lambda_{k}} \ N(t_{0}) \ \mathrm{E}^{\mathbb{Q}^{N}}\left( \sum_{i=0}^{n} \frac{X_{i}}{N(t_{i})} \exp\left( - \Lambda_{i} \right) \right) \bigg\vert_{\Lambda_{k} = 0}
            \ = \
            P(t_{k};t_{0}) \ \mathrm{E}^{\mathbb{Q}^{P(t_{k})}}\left( X_{k} \right)
             \text{.}
        \end{equation*}
    \end{proposition}

    \begin{proof}
       	If $\exp\left( - \int_{0}^{t_{i}} \lambda(s) \ \mathrm{d}s \right)$ models a survival probability of the cash flow $X_{i}$, then the value of the series of cash-flows $X_{i}$ is
        \begin{multline*}
        		N(t_{0}) \ \mathrm{E}^{\mathbb{Q}^{N}}\left( \sum_{i=0}^{n} \frac{X_{i}}{N(t_{i})} \exp\left( - \int_{0}^{t_{i}} \lambda(s) \ \mathrm{d}s \right) \right) \\
        		\ = \
        		N(t_{0}) \ \mathrm{E}^{\mathbb{Q}^{N}}\left( \sum_{i=0}^{n} \frac{X_{i}}{N(t_{i})} \exp\left( - \Lambda_{i} \right) \right) \text{.}                        
        \end{multline*}
        Note that, for the baseline scenario $\Lambda_{i}=0$ for all $i$, we recover the risk‑neutral valuation of the cash-flows.
        From the assumption $\frac{\partial X_{i}}{\partial \Lambda_{k}} = 0$ we then have
        \begin{align*}	    
            & - \frac{\partial}{\partial \Lambda_{k}} \ N(t_{0}) \ \mathrm{E}^{\mathbb{Q}^{N}}\left( \sum_{i=0}^{n} \frac{X_{i}}{N(t_{i})} \exp\left( - \Lambda_{i} \right) \right) \Big\vert_{\Lambda_{k} = 0} \\
            & \qquad \ = \
            N(t_{0}) \ \mathrm{E}^{\mathbb{Q}^{N}}\left( \frac{X_{k}}{N(t_{k})} \exp\left( - \Lambda_{k} \right) \right) \Big\vert_{\Lambda_{k} = 0}
            \ = \
            N(t_{0}) \ \mathrm{E}^{\mathbb{Q}^{N}}\left( \frac{X_{k}}{N(t_{k})} \right) \\
            \intertext{and by choosing $N=P(t_{k})$}
            & \qquad \ = \
            P(t_{k};t_{0}) \ \mathrm{E}^{\mathbb{Q}^{P(t_{k})}}\left( X_{k} \right)
             \text{.}
        \end{align*}
    \end{proof}
    	
	\bigskip

	In the above consideration, it is crucial that the cash flow is independent of the intensity process $\lambda$. The relation breaks down, if the cash-flow depends on $\lambda$.

	\subsection{Cash-Flows with Stochastic Time}

	Let $\tau$ denote a stopping time (a stochastic time). For example, $\tau$ could represent the optimal exercise of a multi-callable product (a Bermudan option). Consider the financial product, where a notional flow of $1$ unit of currency is paid at time $\tau$. Assume that $\tau$ can take only the discrete values from $\{ t_{i}, i=0,\ldots.n \}$. The cash-flow in time $t_{i}$ is then given by the indicator function $X_{i} = \indicatorfcn_{\{\tau = t_{i}\}}$. Then, for any probability measure $\mathbb{Q}$, the total sum of the expected cash-flows is one, i.e., we have
	\begin{equation*}
		\sum_{i=0}^{n} \mathrm{E}^{\mathbb{Q}} \left( X_{i} \right) \ = \ 1 \text{.}
	\end{equation*}
	From the natural notion of \textit{expected cash-flows}, it appears as if it is a natural consistency condition to assume that the separate expected cash-flows of a unit cash flow on a stochastic time add to one, if summed over all possible times.

	However, taking the risk‑neutral valuations
	\begin{equation*}
        E_{i} \ := \ EC^{1/P}\left( X_{i} \right) \ = \ 
        N(t_{0}) \mathrm{E}^{\mathbb{Q}^{N}}\left( \frac{X_{i}}{N(t_{i})} \right)
        \cdot \frac{1}{P(t_{i};t_{0})}
	\end{equation*}
	it is not guaranteed that the sum $E_{i}$ equals $1$. This is due to a possible correlation of the indicator function and the interest rate level.
    \paragraph{Remark (When do the replication-relevant timing weights add up?).}
    Let $\tau\in\{t_1,\ldots,t_n\}$ and $X_{i}=\indicatorfcn_{\{\tau=t_{i}\}}$. Using the bond pricing identity
    \[
    P(t_{i};t_0)=N(t_0)\,\mathrm{E}^{\mathbb{Q}^{N}}\!\left(\frac{1}{N(t_{i})}\right),
    \]
    we can rewrite
    \begin{equation}
    \label{eq:liqui:additivity-condition}
    E_{i}
    =
    \frac{\mathrm{E}^{\mathbb{Q}^{N}}\!\left(\frac{X_{i}}{N(t_{i})}\right)}
    {\mathrm{E}^{\mathbb{Q}^{N}}\!\left(\frac{1}{N(t_{i})}\right)}.
    \end{equation}
    Hence, if $X_{i}$ (equivalently, the stopping time $\tau$) is independent of $\frac{1}{N(t_{i})}$
    under $\mathbb{Q}^{N}$ (no timing--rate dependence), then
    \[
    E_{i}=\mathrm{E}^{\mathbb{Q}^{N}}(X_{i})
    \quad\text{and therefore}\quad
    \sum_{i=1}^{n} E_{i}=1.
    \]
    In general, dependence between timing indicators and discounting/funding implies $\sum_{i=1}^{n}E_{i}\neq 1$.
    
	\subsubsection{Example: Risk-Neutral Valuation of Indicators / Correlation}

	Consider two payment times $T_{1}$ and $T_{2}$ and the payoffs
	\begin{equation*}
		X_{1} \ := \ \indicatorfcn_{\{ L(T_{1},T_{2};T_{1}) > K \}} \quad \text{in $T_{1}$} \qquad \text{and} \qquad
		X_{2} \ := \ 1 - X_{1} \quad \text{in $T_{2}$.}
	\end{equation*}
	The payment of $X_{2}$ is that of a digital floorlet. The payment of $X_{1}$ is that of a digital caplet in arrears. For $E_{1} + E_{2}$ we have
	\begin{align*}
		E_{1} + E_{2} & \ = \ \mathrm{E}^{\mathbb{Q}^{P(T_{1})}}\left( X_{1} \right) + \mathrm{E}^{\mathbb{Q}^{P(T_{2})}}\left( X_{2} \right) \\
		& \ = \ \mathrm{E}^{\mathbb{Q}^{P(T_{2})}}\left( X_{1} \frac{P(T_{1};T_{1})}{P(T_{2};T_{1})} \frac{P(T_{2};t_{0})}{P(T_{1};t_{0})} \right) + \mathrm{E}^{\mathbb{Q}^{P(T_{2})}}\left( X_{2} \right) \\
		& \ = \ \mathrm{E}^{\mathbb{Q}^{P(T_{2})}}\left( X_{1} \frac{1 + L(T_{1},T_{2};T_{1}) (T_{2}-T_{1})}{1 + L(T_{1},T_{2};t_{0}) (T_{2}-T_{1})} \right) + \mathrm{E}^{\mathbb{Q}^{P(T_{2})}}\left( X_{2} \right) \\
		& \ = \ \mathrm{E}^{\mathbb{Q}^{P(T_{2})}}\left( X_{1} \frac{(L(T_{1},T_{2};T_{1})-L(T_{1},T_{2};t_{0})) (T_{2}-T_{1})}{1 + L(T_{1},T_{2};t_{0}) (T_{2}-T_{1})} \right) + \mathrm{E}^{\mathbb{Q}^{P(T_{2})}}\left( X_{1} + X_{2} \right)
	\end{align*}
	Thus, we have
	\begin{align*}
		E_{1} + E_{2} - 1 & \ = \ \mathrm{E}^{\mathbb{Q}^{P(T_{2})}}\left( X_{1} \frac{(L(T_{1},T_{2};T_{1})-L(T_{1},T_{2};t_{0})) (T_{2}-T_{1})}{1 + L(T_{1},T_{2};t_{0}) (T_{2}-T_{1})} \right) \\
		& \ = \ \mathrm{E}^{\mathbb{Q}^{P(T_{2})}}\left( X_{1} (L(T_{1},T_{2};T_{1})-L(T_{1},T_{2};t_{0})) (T_{2}-T_{1}) \right) \frac{P(T_{2};t_{0})}{P(T_{1};t_{0})} \text{.}
	\end{align*}
	Note that $\mathrm{E}^{\mathbb{Q}^{P(T_{2})}}\left( (L(T_{1},T_{2};T_{1})-L(T_{1},T_{2};t_{0})) (T_{2}-T_{1}) \right) = 0$, but choosing $K = L(T_{1},T_{2};t_{0})$ we see that $\mathrm{E}^{\mathbb{Q}^{P(T_{2})}}\left( X_{1} (L(T_{1},T_{2};T_{1})-L(T_{1},T_{2};t_{0})) (T_{2}-T_{1}) \right) > 0$ if there is a positive probability that $L(T_{1},T_{2};T_{1})$ is above $L(T_{1},T_{2};t_{0})$. Obviously, it is the correlation of the indicator $X_{1}$ with the discount rate $L(T_{1},T_{2};T_{1})$ that is leading to $E_{1} + E_{2} \neq 1$.

    \paragraph{Rough Assessment of the Magnitude}

    Assume that under the $T_{2}$-forward measure the forward rate satisfies a Black-type lognormal model, $L(T_{1},T_{2};T_{1}) = L_{0} \ \exp\!\left(-\tfrac12 \sigma_{\text{L}}^{2}T_{1}+\sigma_{\text{L}} W(T_{1}) \right)$, so that
    \begin{equation*}
        \begin{split}
            \mathrm{E}^{\mathbb{Q}^{P(T_{2})}}\!\left(L(T_{1},T_{2};T_{1})\mathbf{1}_{\{L(T_{1},T_{2};T_{1})>K\}}\right) & = L_{0}\Phi(d_{+}), \\
            \mathbb{Q}^{P(T_{2})}\!\left(L(T_{1},T_{2};T_{1})>K\right) & = \Phi(d_{-}),
        \end{split}
    \end{equation*}
    with $d_{\pm} \ := \ \frac{\log(L_{0}/K) \pm \tfrac12 \sigma_{\text{L}}^{2} T_{1}}{\sigma_{\text{L}}\sqrt{T_{1}}}$. We then find
    \begin{align*}
        & E_{1}+E_{2}-1 \ = \
        \frac{L_{0} \cdot (T_{2}-T_{1})}{1 + L_{0} \cdot (T_{2}-T_{1})}\,\bigl(\Phi(d_{+})-\Phi(d_{-})\bigr) \text{,} \\
        \intertext{and, in particular, for $K = L_{0}$ we find}
        & \ = \ 
        \frac{L_{0} \cdot (T_{2}-T_{1})}{1 + L_{0} \cdot (T_{2}-T_{1})}\,\bigl(\Phi(+\frac{\sigma_{\text{L}}}{2}\sqrt{T_{1}})-\Phi(-\frac{\sigma_{\text{L}}}{2}\sqrt{T_{1}})\bigr) \ \approx \ \frac{L_{0} \cdot (T_{2}-T_{1})}{1 + L_{0} \cdot (T_{2}-T_{1})}\, \frac{\sigma_{\text{L}}\sqrt{T_{1}}}{\sqrt{2 \pi}} \text{.}
    \end{align*}
    Consequently, for large volatility we have a relative inconsistency of $\frac{L_{0} \cdot (T_{2}-T_{1})}{1 + L_{0} \cdot (T_{2}-T_{1})}$ that grows towards 100~\% with the time difference.\footnote{
        Under a normal (Bachelier) approximation $L = L_{0} + \sigma_{\text{N}} W(T_{1})$, the ATM inconsistency becomes $\frac{T_{2}-T_{1}}{1 + L_{0} \cdot (T_{2}-T_{1})}\, \frac{\sigma_{\text{N}}\sqrt{T_{1}}}{\sqrt{2 \pi}}$, which agrees to first order with the Black-based expression via $\sigma_{\text{N}} \approx L_{0} \ \sigma_{\text{L}}$.}

	\subsection{Inconsistency of Expected Cash-Flows and Liquidity Risk Management}

	A risk‑neutral valuation is given by the value of the replication portfolio.\footnote{Under assumptions like completeness and absence of arbitrage.} The partial derivative of the valuation with respect to a basis of assets gives the composition of the replication portfolio in terms of the values of these assets.
	
	If we extend the set of zero coupon bonds $\{ P(T_{i}) \ \vert \ i=0,\ldots,n \}$ with assets that are orthogonal (i.e., their sensitivity with respect to zero bonds is zero), then the sensitivity with respect to $P(T_{i})$ represent the number of units to hold in the respective zero coupon bond to replicate the cash-flow.

    The sensitivities of a portfolio with respect to the funding instruments, say here, the zero coupon bonds, gives the initial term-funding requirement for the portfolio.

    From~\eqref{eq:liqui:ecadvalueperbond} it appears as if this term-funding will just match the associated expected cash-flow. However, the two concepts are inconsistent for non-trivial financial derivatives.

    The inconsistency comes from the expected cash-flows being a static view on future model cash-flows while sensitivities are the initial cost of a self-financing trading strategy. The following elementary example illustrates this.
    
	\subsubsection{Example: Hedge Portfolio of a Stochastic Time Cash-Flow}
    \label{sec:liqui:example:twotimetwostatetwobond}

	Consider two future times $t_{1}$, $t_{2}$ and a stochastic payment time $\tau$ taking values from $\{ t_{1}, t_{2} \}$. We assume that at $t_{1}$ (or shortly before) the state of $\tau$ is known, i.e., it is a stopping time. Furthermore, assume that the probability space can attain only two values $\omega_{1}$, $\omega_{2}$ with equal probability and that $\tau(\omega_{1}) = t_{1}$, $\tau(\omega_{2}) = t_{2}$.

	We now consider a payoff of $1$ unit of currency occurring at the stochastic time $\tau$.

	The $\mathbb{P}$-expected cash-flow at $t_{1}$ is $\mathbb{P}(\tau = t_{1}) = 0.5$, the $\mathbb{P}$-expected cash-flow at $t_{2}$ is $\mathbb{P}(\tau = t_{2}) = 0.5$.

	The risk‑neutral replication strategy of the stochastic payoff is a portfolio of a $t_{1}$-Bond $P(t_{1})$ and a $t_{2}$-Bond $P(t_{2})$
	\begin{equation*}
		a P(t_{1};t_{1}, \omega_{1})  +  b P(t_{2};t_{1}, \omega_{1}) \text{.}
	\end{equation*}
	The coefficients are determined by the value of the two bonds in the two possible states $\omega_{1}$, $\omega_{2}$. In state $\omega_{1}$ the replication portfolio should agree with $P(t_{1})$, in state $\omega_{2}$ the replication portfolio should agree with $P(t_{2})$, that is 
	\begin{equation}
        \label{eq:liqui:twostatetwotimetwobondhedge}
        \begin{split}            
    		P(t_{1};t_{1}, \omega_{1}) & \ = \ a P(t_{1};t_{1}, \omega_{1})  +  b P(t_{2};t_{1}, \omega_{1}) \\
    		P(t_{2};t_{1}, \omega_{2}) & \ = \ a P(t_{1};t_{1}, \omega_{2})  +  b P(t_{2};t_{1}, \omega_{2})
            \end{split}
	\end{equation}    
	When the payment time $\tau$ becomes known, the replication portfolio is liquidated and the respective bond is acquired, replicating the desired cash-flow.

	The solution of \eqref{eq:liqui:twostatetwotimetwobondhedge} gives us the composition of the replication portfolio. With $x_{1} = P(t_{1};t_{1}, \omega_{1})$, $x_{2} = P(t_{2};t_{1}, \omega_{1})$,
	$y_{1} = P(t_{1};t_{1}, \omega_{2})$, $y_{2} = P(t_{2};t_{1}, \omega_{2})$
	the solution of \eqref{eq:liqui:twostatetwotimetwobondhedge} is given by
    \begin{equation*}
        a = \frac{ y_2 (x_2 - x_1) }{ y_1 x_2 - y_2 x_1} \text{,} \qquad
        b = \frac{ x_1 (y_1 - y_2) }{ y_1 x_2 - y_2 x_1} \text{.}
    \end{equation*}

	With bond values being for example
	\begin{align}
		x_1 &= 1.0, \quad x_2 = 0.8, \\
		y_1 &= 1.0, \quad y_2 = 0.9
	\end{align}
	we get
	\begin{align}
		a &= \frac{y_2 (x_2 - x_1)}{y_1 x_2 - y_2 x_1}
		= \frac{0.9 \cdot (0.8 - 1.0)}{1.0 \cdot 0.8 - 0.9 \cdot 1.0}
		= \frac{-0.18}{-0.1}
		= 1.8
		\\
		b &= \frac{x_1 (y_1 - y_2)}{y_1 x_2 - y_2 x_1}
		= \frac{1.0 \cdot (0.9 - 0.8)}{1.0 \cdot 0.8 - 0.9 \cdot 1.0}
		= \frac{0.1}{-0.1}
		= -1.0
	\end{align}

	This is a long-short position of the two bonds. The risk‑neutral trading strategy is to buy a bond with maturity $t_{1}$ guaranteeing the cash-flow of $1.8$ in $t_{1}$ and sell a bond with maturity $t_{2}$.
	In the state $\omega_{1}$ we use $0.8$ of this cash flow to buy back the $t_{2}$-bond and are left with $1$ unit of cash flow in $t_{1}$. In the state $\omega_{2}$ we use cash-flow of $1.8$ to buy two units of the $t_{2}$ bond, being left with one such bond that matches the cash-flow in $t_{2}$.

	The $\mathbb{P}$-expected cash-flows of the (dynamic) replication portfolio do not match the $\mathbb{P}$-expected cash-flow of our stochastic payoff, but the (dynamic) trading strategy allows to replicate the cash-flows exactly. More strikingly, although the combined position, a short position of stochastic time cash-flow and a long position of the replication portfolio, produces zero net cash flow in every scenario $\{ \omega_{1}, \omega_{2} \}$, it exhibits non-zero expected cash flows under $\mathbb{P}$. The $\mathbb{P}$-expected cash flows of the combined position are
	\begin{equation*}
		1.8 - 1.0 \cdot \mathbb{P}(\{ \omega_{1} \})  \quad \text{in $t_{1}$} \qquad -1.0 - 1.0 \cdot \mathbb{P}(\{ \omega_{2} \}) \quad \text{in $t_{2}$,}
	\end{equation*}
	so for equal probabilities we would get
	\begin{equation*}
		+ 1.3 \quad \text{in $t_{1}$} \qquad -1.5 \quad \text{in $t_{2}$.}    
	\end{equation*}
	
	Note that in this example, there is no liquidity risk, as long as we can liquidate the $t_{2}$-bond in $t_{1}$. Being unable to liquidate the $t_{2}$-bond is a stress scenario, which we can apply to the replication strategy, but for the baseline scenario we have matched cash-flows despite the non-zero expected cash-flows.
    
    \clearpage
    \section{Funding Sensitivities as Replication-Consistent Liquidity Measures}
    \label{sec:liqui:funding-sensitivities}

    The example of a stochastic time cash-flow shows a fundamental problem in the use of expected cash-flows: The static expected cash-flows ignore the temporal effect of an implicit (risk-neutralizing) trading strategy; they do not account for a trading activity like rolling the $t_{1}$-bond into a $t_{2}$-bond, depending on some event $\tau$, as illustrated in Section~\ref{sec:liqui:example:twotimetwostatetwobond}. Sensitivities, on the other hand, account for trading activities as risk‑neutral replication is defined by matching sensitivities.

    \paragraph{From static expected cash-flow to liquidity rebalancing of \emph{funding hedge ratios}.}
    Section~\ref{sec:liqui:expected-cashflows} took a \emph{contractual-leg} perspective and discussed expected cash-flows
    $(\mathrm{E}^{\mu}(X_{i}))$ as descriptive projections.
    From here on we adopt a \emph{replication} perspective:
    when we refer to \emph{liquidity-relevant cash movements} we mean the cash movements implied by the
    self-financing hedging strategy.
    Consequently, liquidity forecasting is based on \emph{funding hedge ratios} (discounting/funding sensitivities)
    and their rebalancing over the hedge grid, rather than on static undiscounted expectations of contractual legs.
    The next sections make this precise by relating funding sensitivities to funding requirements and by extracting scenario-wise funding increments from a valuation model.

    \medskip

    We consider sensitivities with respect to the \textit{funding spread}. Let $T \mapsto \lambda(T;t)$ denote the funding spread curve for maturities $T$ observed in $t$, i.e., the value of the funding bond with maturity $T_{i}$ is given by
    \begin{equation}
        \label{eq:liqui:fundingbond}
        P^{f}(T_{i};t) \ = \ P^{\circ}(T_{i};t) \ \exp\left( - \int_{t}^{T_{i}} \lambda(s;t) \ \mathrm{d}s \right)
        \ = \ P^{\circ}(T_{i};t) \ \exp\!\left(-\Lambda(t,T_{i})\right) \text{,}
    \end{equation}
    with $\Lambda(t,T_{i}) \ := \ \int\limits_{t}^{T_{i}}\lambda(u;t)\,du$.

    \medskip

    With respect to stochastic cash-flows we may distinguish stochasticity in the size of the cash-flow and stochasticity in the time of the cash-flow. The two aspects behave differently when considering expected future cash-flows in the risk-neutral replication setting.
    
    In other words, we consider the sequence of cash-flows
    \begin{equation}
        X_{i} \ \indicatorfcn_{\{\tau = T_{i}\}} \text{,}
    \end{equation}
    where $X_{i}$ is an $\mathcal{F}_{T_{i}}$-measurable random variable and $\tau$ is a stopping time (a random time) with values in $\{ T_{1}, \ldots, T_{n} \}$.

    \subsection{Stochasticity in Size}

    The stochasticity in size can be analyzed by expectations and quantiles (e.g., valuations under stress scenarios) of the fixed-time cash-flow. There is a consistent definition of expected cash-flow that matches the hedge ratios of a replication portfolio.

    Let $V_{i}$ denote the value of the cash-flow $X_{i}$ occurring at time $T_{i}$ (independent of $\Lambda_{i}$). Using the $T_{i}$-terminal measure $\mathbb{Q}^{P^{f}(T_{i})}$ for the cash-flow $X_{i}$ we find that the sensitivity of the risk-neutral valuation matches the $\mathbb{Q}^{P^{f}(T_{i})}$-expected cash-flow, i.e.,
    \begin{align*}
        -\frac{\partial}{\partial \Lambda_{i}} V_{i}
        & \ = \ -\frac{\partial}{\partial \Lambda_{i}} \ \left( \mathrm{E}^{\mathbb{Q}^{N}} \left( \frac{X_{i}}{N(T_{i})} \right) \cdot N(0) \right) \\
        & \ = \ -\frac{\partial}{\partial \Lambda_{i}} \ \left( \mathrm{E}^{\mathbb{Q}^{P^{f}(T_{i})}} \left( X_{i} \right) \cdot P^{f}(T_{i};0) \right)        
        \ \stackrel{\eqref{eq:liqui:fundingbond}}{=} \ \mathrm{E}^{\mathbb{Q}^{P^{f}(T_{i})}} \left( X_{i} \right) \cdot P^{f}(T_{i};0) \text{.}
    \end{align*}
    Consequently, neutralizing the funding sensitivities will balance out the $\mathbb{Q}^{P^{f}(T_{i})}$-expected cash-flows.
      
    \subsection{Stochasticity in Time}

    With respect to stochasticity in time, assessing expectations of the individual cash-flows ($X_{i}$) can be misleading.
    
    Consider a constant cash-flow $X$ occurring at a stochastic time $\tau \in \{ T_{i}, T_{j} \}$ which depends on the funding spread or funding bond, say
    \begin{equation*}
        \tau \ = \ \indicatorfcn_{\{Z > \Lambda_{i}\}} \ T_{i} + \left(1 - \indicatorfcn_{\{Z > \Lambda_{i}\}} \right) \ T_{j} \text{,}
    \end{equation*}
    where $Z$ is some $\mathcal{F}_{T_{i}}$-measurable random variable, which we assume to be independent of $\Lambda_{i}$, for simplicity.
    
    Let $V$ denote the value of paying $X$ in time $\tau$. For the funding sensitivity of this valuation we find
    \begin{align*}
        -\frac{\partial}{\partial \Lambda_{i}} V \ = \ & X \cdot \mathbb{Q}^{P^{f}(T_{i})} \left( Z > \Lambda_{i} \right) \cdot P^{f}(T_{i};0) \ + \ X \cdot f_{Z}(\Lambda_{i}) \cdot \left( P^{f}(T_{i};0) - P^{f}(T_{j};0) \right) \text{,}
    \end{align*}
    where $f_{Z}$ is the probability density of $Z$ under $\mathbb{Q}^{P^{f}(T_{i})}$.
    
    Here, the first expression corresponds to the $\mathbb{Q}^{P^{f}(T_{i})}$-expected cash-flow: the cash-flow is multiplied with the probability of a payment, then discounted. Hence, the second expression represents a mismatch of the funding risk sensitivity and the expected cash flows. The mismatch is associated with the sensitivity of the cash-flow time on the funding spread: it is the size of the cash-flow multiplied with change in probability (which is the density).

    A risk‑neutral replication strategy will compensate the risk associated with a change of the payment time. Thus, a perfectly dynamically hedged position will show unbalanced expected cash-flows, even though the self-financing trading strategy will result in perfectly matched cash-flows.


    \clearpage
    
    \section{Obtaining Stochastic Future Funding Requirements from a Valuation Model}
    \label{sec:liqui:replication-consistent-future-liquidity-requirement}

    The sensitivities of a risk‑neutral valuation allow to obtain precise information on cash-flows taking into account the trading activity associated with the risk-neutral replication. This is an advantage compared to the calculation of expected cash-flows, which assume that the portfolios remain static. In addition we may obtain meaningful information in scenarios.

    \subsection{Liquidity Frictions}

    A risk‑neutral valuation model makes idealized assumptions, which may be critical: 
    
    \paragraph{Time Resolution:} Using a time resolution of 1 day, the model usually only reveals a daily net cash-flow: an inflow is matched to an outflow.

    An analysis of these flows can be achieved by introducing an additional time-gap between outflows and inflows. This modification can become part of the risk‑neutral valuation model, which then values the replication cost of the time-gapped cash-flow structure.
    The benefit is two-fold: the sensitivities of the modified valuation report the liquidity requirements to bridge the time-gap, and, the valuation change report the associated cost - the valuation difference can be considered to be an \textit{liquidity (risk) valuation adjustment}. See Section~\ref{sec:liqui:definition-lva}.

    \paragraph{Cross-Asset Transaction Netting:} The replication strategy associated with a risk-neutral valuation model usually makes the idealized assumption that the self-financing replication strategy, i.e., the adjustment of the hedge position, occurs instantaneously without frictions, i.e., does not require intermediary liquidity: selling units of one asset, buying units of another is an atomic transaction.
    For some hedge instruments, this assumption is approximately correct. If instruments are cleared on a central counterparty (CCP), the transaction is settled on a netting basis.

    If the assumption is not fulfilled, it is of interest to extract the information on the required hedge-portfolio adjustment. Note that these adjustments may not be associated with any product cash-flow. They are adjustments of a dynamic hedge induced by market changes.
    This shows the relevance of atomic multi-asset settlement / delivery-versus-payment (DvP) in this context, \cite{FriesKohl-Landgraf2023, lee2021atomic}.

    \subsection{Replication-Consistent Future Liquidity Requirements}

    We assume a risk-neutral valuation model that consistently simulates different funding sources, for example,
    \begin{itemize}
        \item the funding zero coupon bonds with different maturities $T$, $t \mapsto P^{f}(T;t)$,

        \item different collateral accounts, $t \mapsto N^{C_{k}}(t)$ \text{.}
    \end{itemize}
    Let $t \mapsto V(t)$ denote the time-$t$ value of a position or portfolio. Here, $V$, $P^{f}(T)$, $N^{C_{k}}$ are stochastic processes.

    We assume that the model evolves on a (possibly fine) time discretization $\{ t_{i} \}$. We assume that the time-discretization contains the times on which adjustment of the hedge portfolio is performed.

We interpret the valuation model as producing (scenario-wise) net funding cash-flow increments over the re-hedging grid.

    With respect to partial derivatives, we assume that the set of funding instruments, e.g., $P^{f}(t_{i})$ has been extended to a basis with hedge-instruments that are orthogonal (i.e., their sensitivity with respect to the funding instruments is zero), such that the sensitivity with respect to $P^{f}(t_{i})$ represents the number of units to hold in the respective zero coupon bond to replicate the cash-flow.

    As stated in Section~\ref{sec:liqui:intro:setup}, let 
    \[
        \phi_{P^{f}(T)}[V](t) := \frac{\partial V(t)}{\partial P^{f}(T;t)},
        \qquad
    \phi_{C_k}[V](t) := \frac{\partial V(t)}{\partial N^{C_k}(t)}.
    \]
    denote the hedge ratios of the position into funding and collateral sources and let
    \begin{equation}
        \tag{\ref{eq:setup:delta}}
        \delta_{A}[V](t_{i})\ :=\ \phi_{A}[V](t_{i})-\phi_{A}[V](t_{i-1}).
    \end{equation}
    denote the required adjustment in the corresponding source ($A=P^{f}(T)$ or $C_k$). For a given set of funding hedge instruments $\mathcal{A}_{i}$, we denote the rebalancing vector by
    \begin{equation}
        r_{i} := (\delta_A[V](t_{i}))_{A\in\mathcal{A}_{i}} \text{.}
    \end{equation}

\subsubsection{Pure Term-Funding}

    We first restrict ourselves to pure term-funding with respect to $\{ \ P^{f}(t_{i}) \ \vert \ i=0,\ldots,n \ \}$ and consider collateral sources later.

    If the fixing of a cash-flow does not depend on a funding source, then for a predictable contractual cash-flow $X_{i}$ at $t_{i}$, i.e., $X_{i} \in \mathcal{F}_{t_{i-1}}$, one may decompose
    \begin{equation*}
        V(t_{i-1}) = \bar{V}_{i}(t_{i-1}) + X_{i} \, P^{f}(t_{i};t_{i-1}) \text{,}
    \end{equation*}
    with $\phi_{P^{f}(t_{i})}[\bar{V}] = 0$, and hence
    \begin{equation}
        \label{eq:liqui:predictable-cashflow-relation}
        \phi_{P^{f}(t_{i})}[V](t_{i-1}) \ = \ X_{i} \text{.}
    \end{equation}
    Thus, in that case, the level vector $\left( \phi_{P^{f}(t_j)}[V](t_{j-1}) \ \vert \ j=i,\ldots,n \right)$ provides the replication-consistent analogue of a static projected cash-flow profile across maturities.

    The identification \eqref{eq:liqui:predictable-cashflow-relation} applies only to the predictable contractual cash-flow component.
    For funding-dependent payoffs, such as a single-curve FRA or an option that depends on the funding forward rate, the shortest-maturity funding position need not correspond to a pure contractual cash-flow, because part of the apparent payoff is already represented by the funding replication across maturities. This is the essence of considering \emph{replication-consistent cash-flow projections}.

    For an implementation based purely on AAD outputs, it is convenient to choose the \emph{funding-desk boundary} and to define the settled cash increment directly from the funding profile $\phi_{P^{f}(t_{j})}[V]$ and its rebalancing.

    The rebalancing vector
    \begin{equation*}
        r_{i}(\omega) = \left( \delta_{P^{f}(t_{j})}[V](t_{i},\omega) \ \vert \ j=i,\ldots,n \right)
    \end{equation*}
    describes the change of this funding profile over $(t_{i-1},t_{i}]$. In particular, the entries with $j > i$ describe the implied trading activity in the remaining term-funding structure.

    Assuming an idealized pure term-funding implementation without pooling or limits, the maturing shortest-maturity funding position contributes the cash amount
    \begin{equation*}
        \phi_{P^{f}(t_{i})}[V](t_{i-1}) \text{,}
    \end{equation*}
    while executing the rebalancing of the longer-dated funding positions generates gross cash settlements
    \begin{equation*}
        -\delta_{P^{f}(t_j)}[V](t_{i}) \, P^{f}(t_j;t_{i}) \text{,} \qquad j=i+1,\ldots,n \text{.}
    \end{equation*}
    We therefore denote by $\kappa_{i}$ the net cash increment settled at time $t_{i}$ under the chosen settlement convention. Under the funding-desk boundary one has
    \begin{equation}
        \label{eq:lva:term-funding-cash-flow-contributors}
        \kappa_{i}
        \ := \
        \phi_{P^{f}(t_{i})}[V](t_{i-1})
        -
        \sum_{j=i+1}^{n} \delta_{P^{f}(t_j)}[V](t_{i}) \, P^{f}(t_j;t_{i})
        \text{,}
    \end{equation}
    which depends only on the AAD-accessible quantities $\phi$ and $\delta$.
    If one chooses a wider strategy boundary that also includes explicit contractual product cash-flows, then those flows may be incorporated on top of \eqref{eq:lva:term-funding-cash-flow-contributors} according to the corresponding settlement convention.
    In more general setups, $\kappa_{i}(\omega)$ aggregates the externally settled cash effects of unsecured funding, cash collateral, repo funding and other funding sources.

    \subsubsection{Settled or Cash-Collateralized Products}

    Additionally, the time-discrete stochastic processes $t_{i} \mapsto \delta_{C_{k}}[V](t_{i})$ provide precise information on hedge-adjustments for a portfolio $V$ that uses the account $C_{k}$ as funding source, where the adjustments are defined by a portfolio-specific settlement valuation. Examples are collateralized portfolios or settle-to-market portfolios.

    \subsubsection{Mixed Funding Sources}

    In an environment with different funding sources and settlement schemes, the mapping from the rebalancing vector $r_{i}$ to the net cash increment $\kappa_{i}$ depends on the chosen strategy boundary and settlement convention. To keep the exposition simple, we treat the different setups separately. They just differ by different rebalancing vectors and settlement mappings.

    \subsubsection{Relation of the Replication-Consistent Cash-Flows to the Valuation}

    The time-discrete stochastic processes $t_{i} \mapsto \kappa_{i}$ and $t_{i} \mapsto r_{i}$ provide precise, replication-consistent cash-flow information. We may use them to define derived quantities. The replication-implied rebalancing vector $r_{i}$ describes the scenario-wise adjustment of the term-funding structure across maturities and funding sources that may be used for treasury and term-funding management, while $\kappa_{i}$ is the relevant object for immediate liquidity management.

    The level vector $\left( \phi_{P^{f}(t_j)}[V](t_{j-1}) \right)_{j=i,\ldots,n}$ provides the replication-consistent analogue of a static projected cash-flow profile, whereas the increment vector $r_{i}$ records the implied trading activity needed to update that profile.
    The two are linked by the desk’s execution and pooling convention: implementing $r_{i}$ generates gross settlement cash-flows, whose pooled / netted residual is $\kappa_{i}$.

    \begin{proposition}[PV of the contractual cash-flow stream]
    \label{prop:liqui:pv-contractual-cashflows}
    Let $(X_{i})_{i=1,\ldots,n}$ denote the contractual cash-flow stream of the position.
    Then
    \begin{equation}
        \label{eq:liqui:pv-contractual-cashflows}
        V(0)
        =
        N(0)\,\mathrm{E}^{\mathbb{Q}^{N}}\!\left(\sum_{i=1}^{n}\frac{X_{i}}{N(t_{i})}\right)
        \text{.}
    \end{equation}
    \end{proposition}

    \begin{proof}
    This is \eqref{eq:setup:streampricing} with $j=0$.
    \end{proof}

    \begin{proposition}[PV of replication-generated net cash increments up to a discrete-hedging residual]
    \label{prop:liqui:pv-kappa-cashflows}
    Consider the pure term-funding setup. For each hedge interval $(t_{i-1},t_{i}]$, let $H_{i}$ denote the frozen one-step delta-replication portfolio selected at $t_{i-1}$, whose funding part is given by the hedge ratios $\phi_{P^{f}(t_{j})}[V](t_{i-1})$, $j=i,\ldots,n$, and assume that
    \[
        H_{i}(t_{i-1}) = V(t_{i-1}),
        \qquad i=1,\ldots,n.
    \]
    Let $\kappa_{i}$ be the externally settled net funding cash increment at $t_{i}$, given by \eqref{eq:lva:term-funding-cash-flow-contributors}. Define the continuation residual by
    \begin{equation}
        \label{eq:liqui:kappa-residual}
        \varepsilon_{i}^{\kappa}
        \ := \
        V(t_{i}) - \bigl(H_{i}(t_{i})-\kappa_{i}\bigr),
        \qquad i=1,\ldots,n-1,
    \end{equation}
    with $\varepsilon_{n}^{\kappa}:=0$.
    If each frozen hedge portfolio $H_{i}$ is self-financing over $(t_{i-1},t_{i}]$ so that its discounted value is a $\mathbb{Q}^{N}$-martingale on that interval, and if $V(t_{n})=0$, then
    \begin{equation}
        \label{eq:liqui:pv-kappa-cashflows}
        N(0)\,\mathrm{E}^{\mathbb{Q}^{N}}\!\left(\sum_{i=1}^{n}\frac{\kappa_{i}}{N(t_{i})}\right)
        =
        V(0)
        +
        N(0)\,\mathrm{E}^{\mathbb{Q}^{N}}\!\left(\sum_{i=1}^{n-1}\frac{\varepsilon_{i}^{\kappa}}{N(t_{i})}\right)
        \text{.}
    \end{equation}
    \end{proposition}

    \begin{proof}
    Since $H_{i}$ is initialized by $H_{i}(t_{i-1})=V(t_{i-1})$ and is self-financing over $(t_{i-1},t_{i}]$, its discounted value is a martingale on that interval, and hence
    \[
        \frac{V(t_{i-1})}{N(t_{i-1})}
        =
        \mathrm{E}^{\mathbb{Q}^{N}}\!\left(
            \frac{H_{i}(t_{i})}{N(t_{i})}
            \,\middle|\,
            \mathcal{F}_{t_{i-1}}
        \right).
    \]
    Using the definition of $\varepsilon_{i}^{\kappa}$, we have
    \[
        H_{i}(t_{i}) = \kappa_{i} + H_{i}(t_{i}) - \kappa_{i} = \kappa_{i} + V(t_{i}) - \varepsilon_{i}^{\kappa}.
    \]
    Therefore
    \[
        \frac{V(t_{i-1})}{N(t_{i-1})}
        =
        \mathrm{E}^{\mathbb{Q}^{N}}\!\left(
            \frac{\kappa_{i}}{N(t_{i})}
            +
            \frac{V(t_{i})}{N(t_{i})}
            -
            \frac{\varepsilon_{i}^{\kappa}}{N(t_{i})}
            \,\middle|\,
            \mathcal{F}_{t_{i-1}}
        \right).
    \]
    Iterating from $i=1$ to $n$ and using $V(t_{n})=0$ yields \eqref{eq:liqui:pv-kappa-cashflows}.
    \end{proof}

    \begin{proposition}[Order of the valuation residual of the time-discrete hedge]
    \label{prop:liqui:kappa-residual-order}
    Let $|\Pi| := \max_{i=0,\ldots,n-1}(t_{i+1}-t_{i})$ denote the mesh of the hedge grid.
    In addition to the assumptions of Proposition~\ref{prop:liqui:pv-kappa-cashflows}, assume that on each interval $(t_{i-1},t_{i}]$ the active state vector of the hedge is an It\^o diffusion with bounded coefficients under $\mathbb{Q}^{N}$, and that there exists a $C^{1,2}$ value function with bounded derivatives up to third order representing $V$. Assume moreover that $H_{i}$ is the frozen first-order hedge of this value function on $(t_{i-1},t_{i}]$.
    Then the valuation residual in \eqref{eq:liqui:pv-kappa-cashflows} satisfies
    \begin{equation}
        \label{eq:liqui:kappa-residual-order}
        N(0)\,\mathrm{E}^{\mathbb{Q}^{N}}\!\left(\sum_{i=1}^{n-1}\frac{\varepsilon_{i}^{\kappa}}{N(t_{i})}\right)
        =
        O(|\Pi|)
        \text{.}
    \end{equation}
    \end{proposition}

    \begin{proof}
    Over one step, the frozen hedge portfolio $H_{i}$ matches the value and first-order sensitivities at $t_{i-1}$.
    Hence a second-order Taylor--It\^o expansion of the continuation value around $t_{i-1}$ shows that the one-step residual $\varepsilon_{i}^{\kappa}$ is generated by second derivatives of the value function with respect to the active state variables, i.e. by gamma-type terms, plus higher-order remainders.
    Since increments of the active state vector are $O_{L^{2}}(\sqrt{t_{i}-t_{i-1}})$, the one-step hedge error is $O(t_{i}-t_{i-1})$ pathwise.
    After discounting and taking conditional expectations under $\mathbb{Q}^{N}$, the first-order drift contribution cancels by the pricing equation for the value function, so that the conditional expectation of the discounted one-step residual is $O((t_{i}-t_{i-1})^{2})$.
    Summing over the grid on a fixed time horizon gives
    \[
        \sum_{i=1}^{n-1} O((t_{i}-t_{i-1})^{2})
        =
        O(|\Pi|),
    \]
    which yields \eqref{eq:liqui:kappa-residual-order}.
    \end{proof}

    \paragraph{Remark.}
    Proposition~\ref{prop:liqui:pv-contractual-cashflows} shows that the contractual cash-flow representation prices the deal exactly.
    Proposition~\ref{prop:liqui:pv-kappa-cashflows} shows that the replication-generated net cash increments $\kappa_{i}$ price the deal up to the discrete-hedging residual $\varepsilon_{i}^{\kappa}$.
    Proposition~\ref{prop:liqui:kappa-residual-order} states that the corresponding \emph{valuation} residual is first order in the mesh of the hedge grid.
    No analogous pricing identity is claimed for $\phi_{i}$ alone, since $\phi_{i}$ omits the rebalancing cash-flows that are part of $\kappa_{i}$.
    This order statement is not pathwise: the accumulated realized hedging error of the discrete hedge is generally of larger stochastic order.

\medskip

\paragraph{Remark:}

For liquidity management, the net external cash increments are the primary objects of interest: $\kappa_{i}(\omega)$ gives the model-implied net cash increment over $(t_{i-1},t_{i}]$ on scenario $\omega$. By contrast, a non-zero future hedge ratio $\phi_{P^{f}(T)}[V](s)(\omega)$ for $s>0$ should not by itself be interpreted as a funding need that must already be secured at time $0$; it is part of the future self-financing replication strategy and the associated cash-flows occur over time.

\subsection{A Liquidity (Risk) Valuation Adjustment (LVA)}
\label{sec:liqui:definition-lva}

    \paragraph{Remark (terminology).}
    We use the term \emph{liquidity valuation adjustment} (LVA) here in a narrow, operational sense: it captures the model-based funding impact of \emph{timing frictions} (e.g.\ settlement lags / time-gaps) applied to the replication-implied net cash increments. This notion is distinct from \emph{market-liquidity} adjustments driven by bid--ask spreads or price impact.

    \medskip

    In the pure term-funding setup discussed above, \eqref{eq:lva:term-funding-cash-flow-contributors} reflects the funding-desk netting assumption that the settled cash increment is described directly by the maturing shortest-maturity funding position together with the settlements generated by rebalancing the remaining term-funding profile.

    However, \eqref{eq:lva:term-funding-cash-flow-contributors} assumes the frictionless benchmark that all components entering the net cash increment can be settled without timing frictions.

    To assess the risk of a mismatch, we consider a settlement lag $\Delta t$ and define a corresponding valuation adjustment to bridge the cash-flows with additional funding.

    For the final value, it is relevant which netting or settlement scheme is assumed. Two aspects are relevant:
    \begin{enumerate}
        \item \label{en:liqui:lva-decomp-products} Netting of in-flows and out-flows across products.

        \item \label{en:liqui:lva-decomp-flows} Netting of a contractual cash-flow and its funding hedge.
    \end{enumerate}
    To address \ref{en:liqui:lva-decomp-products}, we assume that a portfolio has been decomposed into netting sets and that $V$ denotes one such netting set. In the limiting case, $V$ is an individual product.

    To address \ref{en:liqui:lva-decomp-flows}, we distinguish different assumptions on how contributors in \eqref{eq:lva:term-funding-cash-flow-contributors} are netted. We use the notation
    \begin{equation*}
        \left[ x \right]^{+} \ := \ \max( x , 0 ) \text{,} \qquad \left[ x \right]^{-} \ := \ \max( -x , 0 ) \text{,}
    \end{equation*}
    such that $x = x^{+} - x^{-}$.

    \paragraph{Cash-Flow Replication Netting Assumptions:}
    \begin{enumerate}
        \item \label{en:liqui:flow-netting-contractional} \textbf{contractual flow only:}
        \[
            \delta_{i}^{X,+} = \left[ X_{i} \right]^{+},
            \qquad
            \delta_{i}^{X,-} = \left[ X_{i} \right]^{-} \text{,}
        \]

        \item \label{en:liqui:flow-netting-norebalance} \textbf{replication-consistent maturing funding-bucket flow:}
        \[
            \delta_{i}^{\phi,+} = \left[ \phi_{P^{f}(t_{i})}[V](t_{i-1}) \right]^{+},
            \qquad
            \delta_{i}^{\phi,-} = \left[ \phi_{P^{f}(t_{i})}[V](t_{i-1}) \right]^{-} \text{,}
        \]

        \item \label{en:liqui:flow-netting-rebalancing} \textbf{internal rebalancing flows:}
        \[
            \delta_{i}^{\mathrm{reb},\pm}
            =
            \sum_{j=i+1}^{n}
            \left[
                -\delta_{P^{f}(t_{j})}[V](t_{i}) \, P^{f}(t_{j};t_{i})
            \right]^{\pm}
            \text{,}
        \]

        \item \label{en:liqui:flow-netting-perfect} \textbf{perfect replication netting:}
        \[
            \delta_{i}^{\kappa,+} = \left[ \kappa_{i} \right]^{+},
            \qquad
            \delta_{i}^{\kappa,-} = \left[ \kappa_{i} \right]^{-} \text{.}
        \]
    \end{enumerate}

    For each convention
    \[
        a \in \{ X,\phi,\mathrm{reb},\kappa \}
    \]
    we define the corresponding frictionless net cash increment by
    \begin{equation*}
        \delta_{i}^{a,0} \ := \ \delta_{i}^{a,+} - \delta_{i}^{a,-} \text{.}
    \end{equation*}

    To model a settlement lag (time gap) $\Delta t>0$ for \emph{inflows}---i.e.\ positive net cash-flows become available only at $t_{i}+\Delta t$---we define the time-gapped effective net cash increment at $t_{i}$ by
    \begin{equation}
        \label{eq:lva:gappedcashflow}
        \tilde{\delta}_{i}^{a,\Delta t}
        \ := \
        - \delta_{i}^{a,-}
        \ + \
        \delta_{i}^{a,+} \, P^{f}(t_{i}+\Delta t;t_{i}),
        \qquad
        a \in \{ X,\phi,\mathrm{reb},\kappa \} \text{.}
    \end{equation}
    Here, $P^{f}(t_{i}+\Delta t;t_{i})$ is the (short-term) discount factor used to fund the gap from $t_{i}$ to $t_{i}+\Delta t$.\footnote{Depending on the application, one may use a desk-specific funding curve, an unsecured curve, or a conservative proxy.}

    The corresponding frictionless and time-gapped valuations are
    \begin{equation}
        \label{eq:lva:gappedvalue}
        V^{a,0}(0)
        \ := \
        N(0)\,\mathrm{E}^{\mathbb{Q}^{N}}\!\left(\sum_{i=1}^{n}\frac{\delta_{i}^{a,0}}{N(t_{i})}\right),
        \qquad
        V^{a,\Delta t}(0)
        \ := \
        N(0)\,\mathrm{E}^{\mathbb{Q}^{N}}\!\left(\sum_{i=1}^{n}\frac{\tilde{\delta}_{i}^{a,\Delta t}}{N(t_{i})}\right)
        \text{.}
    \end{equation}

    \paragraph{Remark.}
    For $a=X$, Proposition~\ref{prop:liqui:pv-contractual-cashflows} implies $V^{X,0}(0)=V(0)$.
    For $a=\kappa$, Propositions~\ref{prop:liqui:pv-kappa-cashflows} and \ref{prop:liqui:kappa-residual-order} show that
    \[
        V^{\kappa,0}(0)
        =
        V(0)
        +
        N(0)\,\mathrm{E}^{\mathbb{Q}^{N}}\!\left(\sum_{i=1}^{n-1}\frac{\varepsilon_{i}^{\kappa}}{N(t_{i})}\right)
        =
        V(0) + O(|\Pi|)
        \text{,}
    \]
    under the stated smoothness and diffusion assumptions.
    By contrast, the frictionless quantities $V^{\phi,0}(0)$ and $V^{\mathrm{reb},0}(0)$ are convention-dependent funding-desk quantities and need not agree with the contractual deal value $V(0)$.
    The indexed LVA introduced below therefore compares timing-friction costs \emph{within} a fixed cash-flow representation $a$.

    \paragraph{Definition (LVA).}
    For every convention $a \in \{ X,\phi,\mathrm{reb},\kappa \}$, the corresponding liquidity valuation adjustment for settlement lag $\Delta t$ is defined as
    \begin{equation}
        \label{eq:lva:def}
        \mathrm{LVA}^{a,\Delta t}(0)
        \ := \
        V^{a,\Delta t}(0) - V^{a,0}(0)
        \ = \
        N(0)\,\mathrm{E}^{\mathbb{Q}^{N}}\!\left(\sum_{i=1}^{n}\frac{\tilde{\delta}_{i}^{a,\Delta t}-\delta_{i}^{a,0}}{N(t_{i})}\right).
    \end{equation}

    Under non-negative effective funding rates over the gap (so that $P^{f}(t_{i}+\Delta t;t_{i})\le 1$), one has $\mathrm{LVA}^{a,\Delta t}(0)\le 0$ for every convention $a$.

    \paragraph{Remark:} For the different \textit{cash-flow replication netting assumptions}, the LVA of \ref{en:liqui:flow-netting-contractional} measures a settlement gap in positive contractual flows only. The LVA of \ref{en:liqui:flow-netting-norebalance} measures a settlement gap when the replication consistent cash-flow is considered, but term-funding activity is assumed to be frictionless. The LVA of \ref{en:liqui:flow-netting-rebalancing} assumes that the contractual cash-flow and its hedge offset perfectly, but assesses frictions in the remaining term-funding rebalancing. Finally, the LVA of \ref{en:liqui:flow-netting-perfect} measures timing frictions in the fully netted replication-consistent cash increment $\kappa_{i}$.

    \paragraph{Sign and monotonicity.}
    From the pathwise identity
    \begin{equation}
        \label{eq:lva:delta-diff}
        \tilde{\delta}_{i}^{a,\Delta t}-\delta_{i}^{a,0}
        \ = \ 
        \delta_{i}^{a,+}\bigl(P^{f}(t_{i}+\Delta t;t_{i})-1\bigr),
        \qquad
        a \in \{ X,\phi,\mathrm{reb},\kappa \},
    \end{equation}
    \eqref{eq:lva:def} can be written as
    \begin{equation}
        \label{eq:lva:def-plus}
        \mathrm{LVA}^{a,\Delta t}(0)
        \ = \
        N(0)\,\mathrm{E}^{\mathbb{Q}^{N}}\!\left(\sum_{i=1}^{n}\frac{\delta_{i}^{a,+}\bigl(P^{f}(t_{i}+\Delta t;t_{i})-1\bigr)}{N(t_{i})}\right).
    \end{equation}
    If $P^{f}(t_{i}+\Delta t;t_{i})\le 1$ for $\Delta t>0$ (e.g.\ non-negative effective funding rates over the gap),
    then $\tilde{\delta}_{i}^{a,\Delta t} \le \delta_{i}^{a,0}$ pathwise and hence $\mathrm{LVA}^{a,\Delta t}(0)\le 0$.
    Moreover, if $\Delta t\mapsto P^{f}(t_{i}+\Delta t;t_{i})$ is non-increasing, then $\Delta t\mapsto\mathrm{LVA}^{a,\Delta t}(0)$ is non-increasing as well.

\medskip

\paragraph{Small-gap approximation.}
Introduce the (possibly stochastic) \emph{effective annualized funding rate} over the gap by
\begin{equation}
    \label{eq:lva:effrate}
    r^{f}(t_{i},\Delta t)
    \ :=\ -\frac{1}{\Delta t}\log P^{f}(t_{i}+\Delta t;t_{i}),
\end{equation}
so that
\begin{equation*}
    P^{f}(t_{i}+\Delta t;t_{i})=\exp\!\bigl(-r^{f}(t_{i},\Delta t)\,\Delta t\bigr).
\end{equation*}
Using \eqref{eq:lva:def-plus} and $\exp(-x)-1=-x+O(x^2)$, we obtain for small $\Delta t$
\begin{equation}
    \label{eq:lva:firstorder}
    \begin{split}
        \mathrm{LVA}^{a,\Delta t}(0)
        & \ = \  N(0)\,\mathrm{E}^{\mathbb{Q}^{N}}\!\left(\sum_{i=1}^{n}\frac{\delta_{i}^{a,+}\bigl(e^{-r^{f}(t_{i},\Delta t)\Delta t}-1\bigr)}{N(t_{i})}\right) \\
        & \ \approx\
        -\,N(0)\,\mathrm{E}^{\mathbb{Q}^{N}}\!\left(\sum_{i=1}^{n}\frac{\delta_{i}^{a,+}\,r^{f}(t_{i},\Delta t)\,\Delta t}{N(t_{i})}\right),
    \end{split}
\end{equation}
which makes explicit that (to first order) the adjustment scales with the delayed positive net cash increments and the short-term funding rate over the gap.

\paragraph{Remark} The LVA above is a non-linear adjustment like other xVAs~\cite{BrigoLiuPallaviciniSloth2016}, depending on the netting portfolio, i.e., it makes a difference if the cost is calculated on a per-product basis and then aggregated or for a portfolio. Calculation on a portfolio level would correspond to the assumption of a settlement of the net cash-flow.

    \subsection{Implementation}

    The numerical methods that allow for the determination of the quantities $\delta_{P^{f}(T)}(t_{i})$, $\delta_{C_{k}}(t_{i})$ are well known, though demanding in resources. We give an indicative overview of the relevant aspects:

    We assume that the model is implemented as a Monte-Carlo simulation. Usually such a model simulates future cash-flows $X_{i}(\omega_{k})$ occurring in time $t_{i}$ on the path $\omega_{k}$.

    If finer simulation-time and tenor resolution is required than in standard implementations, short-rate models may be preferable. If, at the same time, a richer multi-factor term-structure representation is desired, a time-homogeneous term-structure model can provide both: a multi-factor term-structure specification together with high time-resolution at the short end \cite{Fries2020TermStructureRefinement}.
    
    As $V(t)$ represents a conditional expectation of the future cash-flows after $t$, it may be required to utilize an American Monte-Carlo method (e.g., a regression method, \cite{LongstaffSchwartz2001}) to estimate the conditional expectation, e.g., for an un-collateralized (funded) derivative
    \begin{equation*}
        V(t) \ = \ N^{f}(t) \ \mathrm{E}\left( \sum_{t_{i} > t} \frac{X_{i}}{N^{f}(t_{i})} \ \vert \ \mathcal{F}_{t} \right) \text{,}
    \end{equation*}
    where $N^{f}$ denotes the funding numeraire, such that $P^{f}(T;t) = \mathrm{E}\left( \frac{N^{f}(t)}{N^{f}(T)} \ \vert \ \mathcal{F}_{t} \right)$ represent the funding zero-bonds.

    \medskip
    
    Partial derivatives such as $\frac{\partial V(t)}{\partial P^{f}(T;t)}$ may be obtained by stochastic adjoint algorithmic differentiation~\cite{FriesAutoDiff4MonteCarloForwardSensitivitiesPreprint, FriesAutoDiff4MonteCarlo}.

%
%
%


    \clearpage
    \section{Conclusion}
    
    This note clarifies why undiscounted \emph{expected cash-flows} are often not the right object for liquidity forecasting of derivative portfolios.
    
    \paragraph{Main messages.}
    \begin{itemize}
        \item \textbf{Measure/numeraire dependence.}
        Undiscounted expectations $\mathrm{E}^{\mathbb{Q}^{N}}(X)$ depend on the chosen equivalent martingale measure (numeraire $N$). Mixing such expectations across products or legs can produce artificial residuals under aggregation or netting.
        
        \item \textbf{Stochastic payment times break the “sum-to-one” intuition.}
        For cash-flows occurring at a stopping time (e.g.\ exercise of Bermudan features), the decomposition into time-bucketed indicator cash-flows $X_{i}=\mathbf{1}_{\{\tau=t_{i}\}}$ yields $\sum_{i} \mathrm{E}^{\mathbb{Q}}(X_{i}) = 1$ under any fixed measure, but the corresponding \emph{replication-relevant} quantities (discounting sensitivities / bond hedge ratios) generally do not add up in that way due to correlation between timing indicators and discounting/funding.
        
        \item \textbf{Replication-consistent forecasting uses funding sensitivities.}
        Liquidity-relevant forecasts should be derived from hedge ratios with respect to funding instruments (and their rebalancing), because these quantities encode the self-financing trading strategy that produces intermediate cash movements.
    \end{itemize}
    
    \paragraph{From valuation to liquidity scenarios.}
    A risk-neutral valuation model can be used to extract pathwise (scenario-wise) funding positions $\phi_{P^{f}(T)}[V](t)$ and their increments $\delta_{P^{f}(T)}[V](t_{i})$, providing distributions of future funding adjustments that are directly relevant for operational liquidity planning and for stress testing trading frictions.
    
    \paragraph{Liquidity valuation adjustment.}
    To bridge the gap between idealized self-financing replication and real-world settlement and funding frictions, we propose a simple liquidity valuation adjustment (LVA) based on a modified cash-flow mapping (e.g.\ settlement lags for inflows). The LVA is defined as a difference of two otherwise consistent valuations and can be computed on a trade or, preferably, on a portfolio/netting-set level.
    
    \paragraph{Limitations and outlook.}
    The analysis relies on the internal consistency of the chosen valuation model and on interpreting sensitivities as replication proxies. In incomplete markets (e.g.\ non-hedgeable funding spread components) the sensitivities remain informative as \emph{model-implied funding requirements}, but they should not be over-interpreted as perfectly hedgeable exposures. Extensions include richer settlement conventions (DvP vs non-DvP), cross-asset liquidity constraints, and calibration of LVA parameters (time-gaps, transaction costs) to institution-specific data.

	\clearpage
	\printbibliography

    \clearpage


\clearpage

\section*{Notes}
\addcontentsline{toc}{section}{Notes}

\subsection*{Suggested Citation}

\begin{itemize}
	\item[] \sloppypar \textsc{Fries, Christian P.}: \textit{\articletitle}
    (2025). Available at SSRN: \url{https://ssrn.com/abstract=5514678} or \url{http://dx.doi.org/10.2139/ssrn.5514678}
	\newline
	\url{http://www.christian-fries.de/finmath}
\end{itemize}

\subsection*{Classification}
\begin{small}
	
	\noindent Classification:
	\href{http://www.ams.org/msc/}{MSC-class}: 65Y04 (Primary), 68M07, 68U20.
	\\
	\phantom{Classification:}
	\href{http://www.aeaweb.org/journal/jel_class_system.html}{JEL-class}: G13, C63.
	\\
	\phantom{Classification:}
	\href{http://www.acm.org/class/1998/ccs98.html}{ACM-class 1998}: G.1.4: Quadrature and Numerical Differentiation
%
%
%

	
	
	\noindent Keywords:
	Derivative Valuation,
	Funding Curve,
    Discounting Sensitivities,
    Expected Cash-Flows,
    Liquidity Management
	
	
		
\end{small}

\subsection*{Version History}

\begin{small}
\begin{center}
\begin{tabular}{p{3.8cm}|p{11.0cm}}

\hline
February 25, 2019
&
Funding-sensitivity as a replacement for expected
cash-flows.
\\[0.6ex]

\hline
September 15, 2025
&
Settlement frictions. SSRN Version.
\\[0.6ex]

\hline
October 15, 2025
&
Revision and extension.
\\[0.6ex]

\hline
April 9, 2026
&
Different netting conventions.\\[0.6ex]

\hline

\end{tabular}
\end{center}
\end{small}

\begin{small}
	\vfill
	
	\bigskip
	\centerline{\small\thepage \ pages. \thecpfNumberOfFigures \ figures. \thecpfNumberOfTables \ tables.}
\end{small}

\end{document}